\documentclass[reqno,12pt,a4paper]{amsart}

\usepackage{mathrsfs}
\usepackage{amssymb}
\usepackage{amsfonts}
\usepackage{latexsym}
\usepackage{amsthm}

\usepackage{enumerate}
\usepackage{titletoc}
\numberwithin{equation}{section}

\usepackage{graphicx}
\usepackage{tikz}
\def\lb{\label}

\begin{document}


\renewcommand{\theequation}{\arabic{section}.\arabic{equation}}
\theoremstyle{plain}
\newtheorem{theorem}{\bf Theorem}[section]
\newtheorem{lemma}[theorem]{\bf Lemma}
\newtheorem{corollary}[theorem]{\bf Corollary}
\newtheorem{proposition}[theorem]{\bf Proposition}
\newtheorem{definition}[theorem]{\bf Definition}
\newtheorem*{definition*}{\bf Definition}
\newtheorem*{example}{\bf Example}
\newtheorem*{theorem*}{\bf Theorem}
\theoremstyle{remark}
\newtheorem*{remark}{\bf Remark}

\def\a{\alpha}  \def\cA{{\mathcal A}}     \def\bA{{\bf A}}  \def\mA{{\mathscr A}}
\def\b{\beta}   \def\cB{{\mathcal B}}     \def\bB{{\bf B}}  \def\mB{{\mathscr B}}
\def\g{\gamma}  \def\cC{{\mathcal C}}     \def\bC{{\bf C}}  \def\mC{{\mathscr C}}
\def\G{\Gamma}  \def\cD{{\mathcal D}}     \def\bD{{\bf D}}  \def\mD{{\mathscr D}}
\def\d{\delta}  \def\cE{{\mathcal E}}     \def\bE{{\bf E}}  \def\mE{{\mathscr E}}
\def\D{\Delta}  \def\cF{{\mathcal F}}     \def\bF{{\bf F}}  \def\mF{{\mathscr F}}
\def\c{\chi}    \def\cG{{\mathcal G}}     \def\bG{{\bf G}}  \def\mG{{\mathscr G}}
\def\z{\zeta}   \def\cH{{\mathcal H}}     \def\bH{{\bf H}}  \def\mH{{\mathscr H}}
\def\e{\eta}    \def\cI{{\mathcal I}}     \def\bI{{\bf I}}  \def\mI{{\mathscr I}}
\def\p{\psi}    \def\cJ{{\mathcal J}}     \def\bJ{{\bf J}}  \def\mJ{{\mathscr J}}
\def\vT{\Theta} \def\cK{{\mathcal K}}     \def\bK{{\bf K}}  \def\mK{{\mathscr K}}
\def\k{\kappa}  \def\cL{{\mathcal L}}     \def\bL{{\bf L}}  \def\mL{{\mathscr L}}
\def\l{\lambda} \def\cM{{\mathcal M}}     \def\bM{{\bf M}}  \def\mM{{\mathscr M}}
\def\L{\Lambda} \def\cN{{\mathcal N}}     \def\bN{{\bf N}}  \def\mN{{\mathscr N}}
\def\m{\mu}     \def\cO{{\mathcal O}}     \def\bO{{\bf O}}  \def\mO{{\mathscr O}}
\def\n{\nu}     \def\cP{{\mathcal P}}     \def\bP{{\bf P}}  \def\mP{{\mathscr P}}
\def\r{\varrho} \def\cQ{{\mathcal Q}}     \def\bQ{{\bf Q}}  \def\mQ{{\mathscr Q}}
\def\s{\sigma}  \def\cR{{\mathcal R}}     \def\bR{{\bf R}}  \def\mR{{\mathscr R}}
\def\S{\Sigma}  \def\cS{{\mathcal S}}     \def\bS{{\bf S}}  \def\mS{{\mathscr S}}
\def\t{\tau}    \def\cT{{\mathcal T}}     \def\bT{{\bf T}}  \def\mT{{\mathscr T}}
\def\f{\phi}    \def\cU{{\mathcal U}}     \def\bU{{\bf U}}  \def\mU{{\mathscr U}}
\def\F{\Phi}    \def\cV{{\mathcal V}}     \def\bV{{\bf V}}  \def\mV{{\mathscr V}}
\def\P{\Psi}    \def\cW{{\mathcal W}}     \def\bW{{\bf W}}  \def\mW{{\mathscr W}}
\def\o{\omega}  \def\cX{{\mathcal X}}     \def\bX{{\bf X}}  \def\mX{{\mathscr X}}
\def\x{\xi}     \def\cY{{\mathcal Y}}     \def\bY{{\bf Y}}  \def\mY{{\mathscr Y}}
\def\X{\Xi}     \def\cZ{{\mathcal Z}}     \def\bZ{{\bf Z}}  \def\mZ{{\mathscr Z}}
\def\O{\Omega}

\newcommand{\mc}{\mathscr {c}}

\newcommand{\gA}{\mathfrak{A}}          \newcommand{\ga}{\mathfrak{a}}
\newcommand{\gB}{\mathfrak{B}}          \newcommand{\gb}{\mathfrak{b}}
\newcommand{\gC}{\mathfrak{C}}          \newcommand{\gc}{\mathfrak{c}}
\newcommand{\gD}{\mathfrak{D}}          \newcommand{\gd}{\mathfrak{d}}
\newcommand{\gE}{\mathfrak{E}}
\newcommand{\gF}{\mathfrak{F}}           \newcommand{\gf}{\mathfrak{f}}
\newcommand{\gG}{\mathfrak{G}}           
\newcommand{\gH}{\mathfrak{H}}           \newcommand{\gh}{\mathfrak{h}}
\newcommand{\gI}{\mathfrak{I}}           \newcommand{\gi}{\mathfrak{i}}
\newcommand{\gJ}{\mathfrak{J}}           \newcommand{\gj}{\mathfrak{j}}
\newcommand{\gK}{\mathfrak{K}}            \newcommand{\gk}{\mathfrak{k}}
\newcommand{\gL}{\mathfrak{L}}            \newcommand{\gl}{\mathfrak{l}}
\newcommand{\gM}{\mathfrak{M}}            \newcommand{\gm}{\mathfrak{m}}
\newcommand{\gN}{\mathfrak{N}}            \newcommand{\gn}{\mathfrak{n}}
\newcommand{\gO}{\mathfrak{O}}
\newcommand{\gP}{\mathfrak{P}}             \newcommand{\gp}{\mathfrak{p}}
\newcommand{\gQ}{\mathfrak{Q}}             \newcommand{\gq}{\mathfrak{q}}
\newcommand{\gR}{\mathfrak{R}}             \newcommand{\gr}{\mathfrak{r}}
\newcommand{\gS}{\mathfrak{S}}              \newcommand{\gs}{\mathfrak{s}}
\newcommand{\gT}{\mathfrak{T}}             \newcommand{\gt}{\mathfrak{t}}
\newcommand{\gU}{\mathfrak{U}}             \newcommand{\gu}{\mathfrak{u}}
\newcommand{\gV}{\mathfrak{V}}             \newcommand{\gv}{\mathfrak{v}}
\newcommand{\gW}{\mathfrak{W}}             \newcommand{\gw}{\mathfrak{w}}
\newcommand{\gX}{\mathfrak{X}}               \newcommand{\gx}{\mathfrak{x}}
\newcommand{\gY}{\mathfrak{Y}}              \newcommand{\gy}{\mathfrak{y}}
\newcommand{\gZ}{\mathfrak{Z}}             \newcommand{\gz}{\mathfrak{z}}

\def\ve{\varepsilon}   \def\vt{\vartheta}    \def\vp{\varphi}    \def\vk{\varkappa}

\def\A{{\mathbb A}} \def\B{{\mathbb B}} \def\C{{\mathbb C}}
\def\dD{{\mathbb D}} \def\E{{\mathbb E}} \def\dF{{\mathbb F}} \def\dG{{\mathbb G}}
\def\H{{\mathbb H}}\def\I{{\mathbb I}} \def\J{{\mathbb J}} \def\K{{\mathbb K}} \def\dL{{\mathbb L}}
\def\M{{\mathbb M}} \def\N{{\mathbb N}} \def\O{{\mathbb O}} \def\dP{{\mathbb P}} \def\R{{\mathbb R}}
\def\dQ{{\mathbb Q}} \def\S{{\mathbb S}} \def\T{{\mathbb T}} \def\U{{\mathbb U}} \def\V{{\mathbb V}}
\def\W{{\mathbb W}} \def\X{{\mathbb X}} \def\Y{{\mathbb Y}} \def\Z{{\mathbb Z}}

\newcommand{\1}{\mathbbm 1}
\newcommand{\dd}    {\, \mathrm d}



\def\la{\leftarrow}              \def\ra{\rightarrow}            \def\Ra{\Rightarrow}
\def\ua{\uparrow}                \def\da{\downarrow}
\def\lra{\leftrightarrow}        \def\Lra{\Leftrightarrow}


\def\lt{\biggl}                  \def\rt{\biggr}
\def\ol{\overline}               \def\wt{\widetilde}
\def\no{\noindent}


\let\ge\geqslant                 \let\le\leqslant
\def\lan{\langle}                \def\ran{\rangle}
\def\/{\over}                    \def\iy{\infty}
\def\sm{\setminus}               \def\es{\emptyset}
\def\ss{\subset}                 \def\ts{\times}
\def\pa{\partial}                \def\os{\oplus}
\def\om{\ominus}                 \def\ev{\equiv}
\def\iint{\int\!\!\!\int}        \def\iintt{\mathop{\int\!\!\int\!\!\dots\!\!\int}\limits}
\def\el2{\ell^{\,2}}             \def\1{1\!\!1}
\def\sh{\sharp}
\def\wh{\widehat}
\def\ds{\dotplus}

\def\all{\mathop{\mathrm{all}}\nolimits}
\def\where{\mathop{\mathrm{where}}\nolimits}
\def\as{\mathop{\mathrm{as}}\nolimits}
\def\Area{\mathop{\mathrm{Area}}\nolimits}
\def\arg{\mathop{\mathrm{arg}}\nolimits}
\def\adj{\mathop{\mathrm{adj}}\nolimits}
\def\const{\mathop{\mathrm{const}}\nolimits}
\def\det{\mathop{\mathrm{det}}\nolimits}
\def\diag{\mathop{\mathrm{diag}}\nolimits}
\def\diam{\mathop{\mathrm{diam}}\nolimits}
\def\dim{\mathop{\mathrm{dim}}\nolimits}
\def\dist{\mathop{\mathrm{dist}}\nolimits}
\def\Im{\mathop{\mathrm{Im}}\nolimits}
\def\Iso{\mathop{\mathrm{Iso}}\nolimits}
\def\Ker{\mathop{\mathrm{Ker}}\nolimits}
\def\Lip{\mathop{\mathrm{Lip}}\nolimits}
\def\rank{\mathop{\mathrm{rank}}\limits}
\def\Ran{\mathop{\mathrm{Ran}}\nolimits}
\def\Re{\mathop{\mathrm{Re}}\nolimits}
\def\Res{\mathop{\mathrm{Res}}\nolimits}
\def\res{\mathop{\mathrm{res}}\limits}
\def\sign{\mathop{\mathrm{sign}}\nolimits}
\def\supp{\mathop{\mathrm{supp}}\nolimits}
\def\Tr{\mathop{\mathrm{Tr}}\nolimits}
\def\AC{\mathop{\rm AC}\nolimits}
\def\BBox{\hspace{1mm}\vrule height6pt width5.5pt depth0pt \hspace{6pt}}


\newcommand\nh[2]{\widehat{#1}\vphantom{#1}^{(#2)}}
\def\dia{\diamond}

\def\Oplus{\bigoplus\nolimits}




\def\qqq{\qquad}
\def\qq{\quad}
\let\ge\geqslant
\let\le\leqslant
\let\geq\geqslant
\let\leq\leqslant

\newcommand{\ca}{\begin{cases}}
\newcommand{\ac}{\end{cases}}
\newcommand{\ma}{\begin{pmatrix}}
\newcommand{\am}{\end{pmatrix}}
\renewcommand{\[}{\begin{equation}}
\renewcommand{\]}{\end{equation}}
\def\bu{\bullet}

\title[{}]
{Inverse resonance scattering for Dirac operators on the half-line}

\date{\today}

\author[Evgeny Korotyaev]{Evgeny Korotyaev}
\address{Department of Analysis,  Saint Petersburg State University,   Universitetskaya nab. 7/9,
St. Petersburg, 199034, Russia, \ korotyaev@gmail.com, \ e.korotyaev@spbu.ru}
\author[Dmitrii Mokeev]{Dmitrii Mokeev}
\address{Saint Petersburg State University,   Universitetskaya nab. 7/9, St.
Petersburg, 199034, Russia, \ mokeev.ds@yandex.ru}

\subjclass{} \keywords{Dirac operators, inverse problems, resonances, canonical systems, compactly supported potentials}

\begin{abstract}
    We consider massless Dirac operators on the half-line with compactly supported potentials.
    We solve the inverse problems in terms of Jost function and scattering matrix
    (including characterization).
    We study resonances as zeros of Jost function and prove that a potential is
    uniquely determined by its resonances. Moreover, we prove the following:

    1)  resonances are
    free parameters and a potential continuously depends on a resonance,

    2) the forbidden domain for resonances is estimated,

    3) asymptotics of resonance counting function is determined,

    4) these results are applied to canonical systems.
\end{abstract}

\maketitle



{\it \footnotesize Dedicated to the memory of Alexey Borisovich Shabat (1937-2020)}

\section{Introduction and main results} \label{p2}

\subsection{Introduction}
We consider inverse problem for Dirac operators on the half-line
with compactly supported potentials. Such operators have many
physical and mathematical applications. In particular, they arise in
study of stationary Dirac equations in $\R^3$ with spherically
symmetric potentials (see e.g. \cite{T92}). Dirac operators are also
known as $2\times2$ Zakharov-Shabat (or AKNS) system, which were
used by Zakharov and Shabat \cite{ZS71} to study nonlinear
Schr{\"o}dinger equation (see also \cite{APT04, DEGM82, FT07}).
In our paper, we consider the self-adjoint Dirac operator $H = H_{\a}$ on $L^2(\R_+,\C^2)$
given by
\[ \label{intro:operator}
    H_{\a} y = -i \s_3 y' + i \s_3 Q y,\qq y = \ma y_1 \\ y_2 \am,\qq \s_3 = \ma 1 & 0 \\ 0 & -1 \am,
\]
with the boundary condition
\[ \label{intro:bc}
    e^{-i\a} y_1(0) - e^{i\a} y_2(0) = 0,\qq \a \in [0,\pi),
\]
where the parameter $\a$ is fixed throughout this paper. Note that
if $\a = 0$, then (\ref{intro:bc}) is the Dirichlet boundary
condition, and if $\a = \frac{\pi}{2}$, then (\ref{intro:bc}) is the
Neumann boundary condition. The potential $Q$ has the following form
\[ \label{intro:potential}
    Q = \ma 0 & q \\ \overline{q} & 0 \am,\qq q \in \cP,
\]
where the class $\cP$ is defined  for some $\g > 0$ fixed throughout
this paper by
\begin{definition*}
    $\cP$ is the set of all functions $q \in L^2(\R_+)$ such that
    $\sup \supp q = \g$.
\end{definition*}
It is well known that
$\s(H) = \s_{ac}(H) = \R$ (see e.g. \cite{LS91}).
We introduce the $2 \times 2$ matrix-valued Jost solution $f(x,z) =  \left(
\begin{smallmatrix} f_{11} & f_{12} \\ f_{21} & f_{22} \end{smallmatrix} \right) (x,z)$
of the Dirac equation
\[ \label{intro:equation}
    f'(x,z) = Q(x) f(x,z) + i z \s_3 f(x,z),\qq (x,z) \in \R_+ \ts \C,
\]
which satisfies the standard condition for compactly supported potentials:
$$
    f(x,z) = e^{i z x \s_3},\qq \forall \qq (x,z) \in [\g,+\iy) \ts \C.
$$
We define a Jost function $\psi: \C \to \C$ by
\[ \label{p2e10}
    \psi(z) = \psi_{\a}(z) = e^{-i\a} f_{11} (0,z) - e^{i\a} f_{21} (0,z),\qq z \in \C.
\]
It is well-known that $\psi$ is entire, $\p(z) \neq 0$ for any $z
\in \ol \C_+ $ and it has zeros in $\C_-$, which are called
\textit{resonances} and a multiplicity of a resonance is a multiplicity of the zero of $\p$.
The resonances are also zeros of the Fredholm
determinant and poles of the resolvent of the operator $H$ (see e.g.
\cite{IK14b}). We also define a scattering matrix $S:\R \to \C$ by
\[ \label{p2e3}
    S(z) = \frac{\ol\psi(z)}{\psi(z)} = e^{-2 i \arg \psi(z)},\qq z \in \R.
\]
The function $S$ admits a meromorphic continuation from $\R$ onto $\C$, since $\psi$ is entire.
Moreover, poles of $S$ are zeros of $\psi$ and then they are resonances. We sometimes write
$\psi(\cdot,q)$, $S(\cdot,q)$, $\ldots$ instead of $\psi(\cdot)$, $S(\cdot)$, $\ldots$,
when several potentials are being dealt with.

Our main goal is to solve inverse problems for the Dirac operator
$H$ with different spectral data: \textit{the scattering matrix},
\textit{the Jost function}, and \textit{the resonances}. In general,
an inverse problem is to determine the potential by some data, and
it consists at least of the four parts:
\begin{enumerate}[(i)]
    \item {\it Uniqueness.} Do data uniquely determine the potential?
    \item {\it Reconstruction.} Give an algorithm to recover the potential by data.
    \item {\it Characterization.} Give necessary and sufficient conditions that data correspond
    to a potential.
    \item {\it Continuity.} Is a potential a continuous function of data and how can
    data be changed so that they remain data for some potential?
\end{enumerate}

We solve the parts (i) -- (iii) for the inverse resonance problem, when data are
resonances. In order to get these results, we solve the inverse problem for the operator $H$
in terms of the Jost function and prove that it
is uniquely determined and recovered by its zeros, i.e. by resonances of $H$.
We also solve the part (iv) for a finite number of resonances. Firstly, we show that a
resonance of $H$ is a free parameter, i.e. if we arbitrarily shift a zero of a Jost function,
then we obtain a Jost function for some potential from $\cP$. Secondly, we prove that a potential
continuously depends on one resonance, where all other resonances are fixed.
Thirdly, we show that if we arbitrarily shift all
zeros of a Jost function along the real line or reflect they across the imaginary line,
then we obtain a Jost function for some potential from $\cP$.
Note that we solve the inverse problem for
the operator $H$ in terms of the Jost function by using the relation
between the scattering matrix and the Jost function and solution of the inverse scattering problem
for compactly supported potential.

In our paper, we use the methods from the paper \cite{K04a}, where
the same problem for the Schr{\"o}dinger operator on the half-line was
considered, and the solution of the inverse scattering problem for
not necessarily compactly supported potentials (see e.g. \cite{APT04,HM16}).
However, there exist differences between Dirac and Schr{\"o}dinger cases,
which require an adaptation of the proofs.
We describe the main differences between them:

\begin{enumerate}[(i)]
    \item The resonances of Dirac operators are not symmetric with respect to the imaginary
    line.

    \item Roughly speaking, the spectral problem for Dirac operators
    corresponds to spectral problem for Schr{\"o}dinger operators with
    distributions.

    \item The second term in the asymptotic expansion of the Jost function of Dirac operators
    decrease more slowly as spectral parameter goes to infinity. Maybe it is the
    main point.
\end{enumerate}

There are a lot papers about resonances in the different setting,
see articles \cite{F97, H99, K04a, S00, Z87} and the book
 \cite{DZ19} and the references therein. The inverse resonance problem for
Schr{\"o}dinger operators with compactly supported potentials was
solved in \cite{K05} for the case of the real line and in
\cite{K04a} for the case of the half line. In these papers, the
uniqueness, reconstruction, and characterization problems were
solved, see also Zworski \cite{Z02}, Brown-Knowles-Weikard
\cite{BKW03} concerning the uniqueness.
Moreover, there are other results about
perturbations of the following model (unperturbed) potentials by
compactly supported potentials: step potentials \cite{C06}, periodic
potentials \cite{K11h}, and linear potentials (corresponding to
one-dimensional Stark operators) \cite{K17}.

In the theory of resonances, one of the basic result is the
asymptotics of the counting function of resonances, which is an
analogue of the Weyl law for eigenvalues. For Schr{\"o}dinger
operators on the real line with compactly supported potentials, such
result was first obtained by Zworski in \cite{Z87}. The "local
resonance" stability problems were considered in \cite{K04b, MSW10}
and results about the Carleson measures for resonances
were obtained in \cite{K16}.

In our paper, we discuss the inverse resonance problem for Dirac
operators on the half-line. As far as we know, this problem has not
been studied enough. Now, we shortly discuss the known results on
the resonances of one-dimensional Dirac operators. Global estimates
of resonances for the massless Dirac operators on the real line were
obtained in \cite{K14}. Resonances for Dirac operators was also
studied in \cite{IK14b} for the massive Dirac operators on the
half-line and in \cite{IK14a} for the massless Dirac operators on
the real line under the condition $q' \in L^1(\R)$. In these papers,
the following results were obtained:
\begin{enumerate}[(i)]
    \item asymptotics of counting function of the resonances;
    \item estimates on the resonances and the forbidden domain;
    \item the trace formula in terms of resonances for the massless case.
\end{enumerate}
In \cite{IK15}, the radial Dirac operator was considered. There is a
number of papers dealing with other related problems for the
one-dimensional Dirac operators, for instance, the resonances for
Dirac fields in black holes was described, see e.g., \cite{I18}.

Note that Dirac operators can be rewritten as canonical systems (see
e.g. p.~389 in \cite{GK67}). For these systems, the inverse problem
can be solved in terms of de Branges spaces (see \cite{dB, R14}).
There exist many papers devoted to de Branges spaces and canonical
systems. In particular, they are used in the inverse spectral theory
of Schr{\"o}dinger and Dirac operators (see e.g. \cite{R02}). It is
well-known that there exist the connection between Jost solutions
and de Branges spaces. Using this fact, we give the characterization
of de Brange spaces associated with the Dirac operators. Similar
characterization in case of the Schr{\"o}dinger operators was given
in \cite{BBP} (see also \cite{P}). We also describe canonical
systems, which are associated with Dirac operators.

\subsection{Main results}
We introduce the class of all Jost functions.
\begin{definition*}
    $(\cJ, \rho_{\cJ})$ is a metric space, where $\cJ$ is the set of all entire
    functions $\psi$ such that
    \[ \label{p2e1}
        \psi(z) = e^{-i\a}  + \int_0^{\g} g(s) e^{2 i z s} ds,\qq z \in \C,
    \]
    for some $g \in \cP$ and $\psi(z) \neq 0$ for any $z \in \ol \C_+$, and
    the metric $\rho_{\cJ}$ is given by
    \[ \label{p2e7}
        \rho_{\cJ}(\psi_1,\psi_2) = \|g_1 - g_2\|_{L^2(0,\g)},\qq \psi_1,\psi_2 \in \cJ.
    \]
\end{definition*}
\begin{remark}
    This class is similar to the case of Schr{\"o}dinger operators from \cite{K04a}, but
    there are the following differences from the case of Schr{\"o}dinger operators:
    \begin{enumerate}[(i)]
        \item there are no zeros in $\C_+$, since we consider the massless case;
        \item zeros in $\C_-$ are not symmetric with respect to the imaginary line;
        \item $\psi(z) - e^{-i\a}$ decreases more slowly as $z \to \pm \iy$.
    \end{enumerate}
\end{remark}
We define the circle $\S^1 = \{\,z \in \C \, \mid \, |z| = 1 \, \}$.
Let $g:\R \to \S^1$ be a continuous function
such that $g(x) = C + o(1)$ as $x \to \pm \iy$ for some $C \in \S^1$.
Then $g = e^{-2i\phi}$ for some continuous $\phi:\R \to \R$.
We introduce a winding number $W(g) \in \Z$ by
$$
    W(g) = \frac{1}{\pi}\left(\lim_{x \to +\iy} \phi(x) - \lim_{x \to -\iy} \phi(x) \right),
$$
i.e., $W(g)$ is a number of revolutions of $g(x)$ around $0$, when
$x$ runs through $\R$. We introduce a class of the scattering
matrices by
\begin{definition*}
    $(\cS,\rho_{\cS})$ is a metric space, where $\cS$ is the set of all
    continuous functions $S:\R \to \S^1$ such that $W(S) = 0$ and there exist
    $F \in L^1(\R) \cap L^2(\R)$ such that $\inf \supp F = -\g$ and
    \[ \label{p2e4}
        S(z) = e^{2i\a}  + \int_{-\g}^{+\iy} F(s) e^{2 i z s} ds,\qq z \in \R;
    \]
    and the metric $\rho_{\cS}$ is given by
    $$
        \rho_{\cS}(S_1,S_2) = \|F_1 - F_2\|_{L^2(-\g,+\iy)} + \|F_1 - F_2\|_{L^1(-\g,+\iy)},\qq
        S_1,S_2 \in \cS.
    $$
\end{definition*}
\begin{remark}
    It follows from well-known properties of the Fourier transform that $S$ given by (\ref{p2e4})
    is continuous and $S(x) =e^{2i\a}+o(1)$ as $x \to \pm\iy$. Thus, the
    winding number $W(S)$ is correctly defined.
\end{remark}
Note that the metric spaces $(\cS,\rho_{\cS})$ and $(\cJ, \rho_{\cJ})$ are not complete.
Moreover, we equip the class $\cP$ with the metric $\rho_{\cP}$ given by
$\rho_{\cP}(q_1,q_2) = \|q_1 - q_2\|_{L^2(0,\g)}$, for any $q_1,q_2 \in \cP$.
Thus, $(\cP,\rho_{\cP})$ is a metric space, which is not complete.
Now, we present our first result.

\begin{theorem} \label{t1}
    \begin{enumerate}[i)]
        \item The mapping $q \mapsto S(\cdot,q)$ from $\cP$ to $\cS$ is a homeomorphism;
        \item The mapping $q \mapsto \psi(\cdot,q)$ from $\cP$ to $\cJ$ is a homeomorphism.
        \item The following identity holds true:
            $$
                \cS = \{ \, S(z) = \ol{\psi(\ol z)} \psi^{-1}(z),\, z \in \C \, \mid \, \psi \in \cJ \, \}.
            $$
    \end{enumerate}
\end{theorem}
\begin{remark}
    We describe an algorithm to recover a potential from the Jost function or
    the scattering matrix in Section \ref{hl}.
\end{remark}

We recall well-known facts about entire functions (see e.g. \cite{Koo98}).
An entire function $g(z)$ is said to be of
\textit{exponential type} if there exist constants $\t,C > 0$ such that $|g(z)| \leq C e^{\t |z|}$,
$z \in \C$.
We introduce a Cartwright class of entire functions by
\begin{definition*}
    $\cE_{Cart}$ is the class of entire functions of exponential type $g$ such that
    \[ \label{p2e8}
        \int_{\R} \frac{\log(1+|g(x)|)dx}{1 + x^2} < \iy,\qq \t_+(g) = 0,\qq \t_-(g) = 2\g,
    \]
    where $\t_{\pm}(g) = \lim \sup_{y \to +\iy} \frac{\log |g(\pm i y)|}{y}$.
\end{definition*}
Let $g \in \cE_{Cart}$ and let $g(0) \neq 0$. We denote by $z_n$, $n
\geq 1$, the zeros of $g$ counted with multiplicity and arranged
that $0 < |z_1| \leq |z_2| \leq \ldots$. Then $g$ has the Hadamard
factorization
\[ \label{p2e13}
    g(z) = g(0) e^{i \g z} \lim_{r \to +\iy}
        \prod_{|z_n| \leq r} \left(1 - \frac{z}{z_n}\right),\qq z \in \C,
\]
see, e.g., p.130 in \cite{L96}, where the product converges
uniformly on compact subsets of $\C$ and
\[ \label{p2e14}
    \sum_{n \geq 1} \frac{|\Im z_n|}{|z_n|^2} < +\iy.
\]
For any entire function $g$ and $(r,\d) \in \R_+ \ts [0,{\pi\/2}]$,
we introduce the following counting functions
$$
    N_{\pm}(r,\d,g) = \# \{ \, z \in \C \, \mid \,
    g(z) = 0,\, |z| \leq r,\,\pm \Re z \geq 0,\, \d < |\arg z| < \pi-\d  \}.
$$
We need the Levinson result about zeros of functions from
$\cE_{Cart}$, see, e.g.,  p. 58 in \cite{Koo98}.

\begin{theorem*}[Levinson] \label{c1}
    Let $g \in \cE_{Cart}$. Then for each $\d > 0$ we have
    \[ \label{p2e15}
        N_{\pm}(r,0,g) = \frac{\g}{\pi} r + o(r),\qq N_{\pm}(r,\d,g) = o(r),
    \]
    as $r \to +\iy$.
\end{theorem*}
It follows from the Paley-Wiener theorem (see e.g. p.30 in \cite{Koo98}), that an entire
function having form (\ref{p2e1}) belongs to the Cartwright class. Thus, we get the following corollary.
\begin{corollary} \label{p1c2}
    Let $q \in \cP$. Then $q$ is uniquely determined by its resonances, $\psi(\cdot,q) \in \cE_{Cart}$
    and it satisfies (\ref{p2e13}-13).
\end{corollary}

    In Corollary \ref{p1c2}, we improve the result from \cite{IK14b}, where a similar result
    was obtained for differentiable potentials.
We describe the position of resonances and the forbidden domain.

\begin{theorem} \label{t5}
    Let $q \in \cP$ and let $z_n$, $n \geq 1$, be its resonances. Let $\ve > 0$. Then there exists
    a constant $C = C(\ve,q) \geq 0$ such that the following inequality holds true for each $n \geq 1$:
    \[ \label{p1e1}
        2 \g \Im z_n \leq \ln \left( \ve + \frac{C}{|z_n|} \right).
    \]
    In particular, for any $A > 0$, there are only finitely many resonances in the strip
    \[ \label{p1e2}
        \{ \, z \in \C \, \mid \, 0 > \Im z > -A \, \}.
    \]
\end{theorem}
\begin{remark}
    If $q' \in L^1(\R_+)$, then estimate (\ref{p1e1}) and the forbidden domain (\ref{p1e2})
    can be given in more detailed form (see Theorem 2.7 in \cite{IK14b}).
\end{remark}

Now, we show how the scattering matrix can be constructed directly by resonances.

\begin{theorem} \label{t3}
    Let $q \in \cP$. Then its scattering matrix $S$ has the following form
    \[ \label{p2e9}
        S(z) = e^{-2i\phi_{sc}(z)},\qq z \in \R,
    \]
    where $\phi_{sc}$ is a real-valued function such that $\phi_{sc} \in L^{\iy}(\R) \cap C^{\iy}(\R)$,
    $\phi_{sc}(\cdot) + \a \in L^2(\R)$, and
    $$
        \phi_{sc}(z) \to -\a \qq \text{as $z \to \pm \iy$}.
    $$
    Moreover, let $z_n$, $n \geq 1$, be zeros of $\psi(\cdot,q)$. Then we have
    \[ \label{p2e6}
        \begin{aligned}
            \phi_{sc}(z) &= \phi_{sc}(0) + \int_{0}^{z} \phi'_{sc}(s) ds,\qq
            \phi'_{sc}(z) = \g + \sum_{n \geq 1} \frac{\Im z_n}{|z - z_n|^2},\qq z \in \R,\\
            \phi_{sc}(0) &= -\a - \lim_{z \to +\iy} \int_{0}^{z} \phi'_{sc}(s) ds =
            -\a + \lim_{z \to -\iy} \int_{z}^{0} \phi'_{sc}(s) ds,
        \end{aligned}
    \]
    where the sum converges absolutely and uniformly on compact subsets of $\R$.
\end{theorem}

We describe some automorphisms of the class $\cJ$. Firstly, we show that the resonances
are free parameters and prove that the Jost function continuously depends on a resonance.
\begin{theorem} \label{t4}
    Let $q^o \in \cP$ and let $z_n^o \in \C_-$, $n \geq 1$, be its
    resonances. Let $N = \# \{ z_j^o \mid |z_j^o| < r \}$ for some $r>1$.
    Let $z_j \in \C_-$, $|z_j|<r$, $j=1,\ldots,N$. Then there exists a unique $q \in \cP$ such~that
    \[ \lb{B1}
        \psi(z,q) = \p(z,q^o)\prod_{j = 1}^N\frac{z-z_j}{z-z_j^o},\qq z \in \C.
    \]
    In particular, any point on $\C_-$ can be a resonance with any
    multiplicity for some $q\in \cP$. Moreover, if each $z_j \to z_j^o$,
    $j=1,\ldots,N$, then we have
    $$
        \rho_{\cP}(q^o, q) \to 0,\qq \rho_{\cJ}(\p(\cdot,q^o), \psi(\cdot,q)) \to 0,\qq
        \rho_{\cS}(S(\cdot,q^o), S(\cdot,q)) \to 0.
    $$
\end{theorem}
\begin{remark}
    For Schr{\"o}dinger operators, similar results are obtained in \cite{K04a}.
\end{remark}

Secondly, we show that $\cJ$ and $\cS$ are invariant with respect to shifts along
the real line and with respect to the reflection across the imaginary line.
\begin{theorem} \label{t6}
    Let $q \in \cP$. Then the following identities hold true:
    \[ \label{p1e3}
        \begin{aligned}
            \psi(z + k,q) &= \psi(z,e_k q),\qq \ol{\psi(-\ol{z},q)} = e^{2i\a}\psi(z,e^{4 i \a}q),\qq z \in \C,\\
            S(z + k,q) &= S(z,e_k q),\qq \ol{S(-\ol{z},q)} = e^{-4i\a}S(z,e^{4 i \a}q),\qq z \in \R,
        \end{aligned}
    \]
    where $e_k(x) = e^{2 i x k}$ for any $k \in \R$.
\end{theorem}
\begin{remark}
    1) Due to (\ref{p1e3}), we can arbitrarily shift all resonances along the real line.

    2) If $f$ is a meromorphic or entire function, then the function $h$ such
    that $h(z) = \ol{f(-\ol{z})}$ is also meromorphic or entire and its zeros are zeros of $f$
    reflected across the imaginary line. Thus, using (\ref{p1e3}), we can reflect all resonances
    across the imaginary line.

    3) For Schr{\"o}dinger operators, such theorem does not hold. It follows from the fact that
    the resonances of Schr{\"o}dinger operators are symmetric with respect to the imaginary line.

    4) If we solve the inverse problem for $\psi \in \cJ$, then, by (\ref{p1e3}), we can solve the inverse problem
    for $\psi(\cdot+k)$ for any $k \in \R$ and for $\psi_1$ such that $\psi_1(z) = \ol{\psi(-\ol{z})}$, $z \in \C$.
\end{remark}

\subsection{Canonical systems}
Now, we apply these results to canonical systems given by
\[ \label{p1e7}
    Jy'(x,z) = z \cH(x) y(x,z),\qq (x,z) \in \R_+ \ts \C,\qq J = \ma 0 & 1 \\ -1 & 0 \am,
\]
where $\cH: \R_+ \to \cM^+_2(\R)$ is a Hamiltonian and by $\cM^+_2(\R)$ we denote the set of
$2 \times 2$ positive-definite self-adjoint matrices with real entries.
The canonical system (\ref{p1e7}) corresponds to an self-adjoint operator
$H_{\cH} = \cH^{-1} J \frac{d}{dx}$ in the weighted Hilbert space $L^2(\R_+, \C^2, \cH)$
equipped with the norm
$$
    \|f\|^2_{L^2(\R_+, \C^2, \cH)} = \int_{\R_+} (f(x), \cH(x) f(x)) dx,\qq f \in L^2(\R_+, \C^2, \cH),
$$
where $(\cdot,\cdot)$ is the standard scalar product in $\C^2$ (see
e.g. \cite{R14}). It is known that the Dirac equation can be written
as a canonical system (see e.g. \cite{R14}). In order to describe
this result, it is convenient to deal with another form of the Dirac
operator. Recall that $H$ is the Dirac operator given by
(\ref{intro:operator}) for some $q \in \cP$. Using the unitary
transformation $T = \frac{1}{\sqrt{2}} \left( \begin{smallmatrix} i
& -i \\ 1 & 1 \end{smallmatrix} \right)$, we construct the Dirac
operator $H_D = T^* H T$, which has the following form
$$
    H_D = J \frac{d}{dx} + V_q,
$$
where
\[ \label{p1e5}
    V_q = \ma q_1 & q_2 \\ q_2 & -q_1 \am,\qq q = -q_2 + i q_1.
\]
We have a similar connection between solutions of the Dirac
equations. Let $y(x,z)$ be a solution of equation
(\ref{intro:equation}) for some $q \in \cP$. Then $u(x,z) = Ty(x,z)$
is a solution of the equation
\[ \label{p1e4}
    J u'(x,z) + V_q(x) u(x,z) = z u(x,z),\qq (x,z) \in \R_+ \ts \C.
\]
\begin{remark}
    These results hold true for $q$ in large class.
\end{remark}
We introduce a fundamental $2 \times 2$ matrix-valued solution $M(x,z,q)$ of equation
(\ref{p1e4}) with potential $V_q$ satisfying the initial condition $M(0,z,q) = I_2$,
where $I_2$ is the $2 \times 2$ identity matrix.
Let $r(x,q) = M(x,0,q)$, $x \in \R_+$, and let $y(x,z) = r^{-1}(x,q) M(x,z,q)$,
$(x,z) \in \R_+ \ts \C$. Then $y(x,z)$ is a solution of the canonical system (\ref{p1e7}) with the
Hamiltonian $\cH_{q} = r^*(\cdot,q) r(\cdot,q)$. Moreover, the Dirac operator $H_D$ with potential
$V_q$, $q \in \cP$, is unitary equivalent to the operator $H_{\cH_q}$ with the unitary
transform $U: L^2(\R_+, \C^2, \cH_q) \to L^2(\R_+,\C^2)$ such that $Uf = r(\cdot,q)f$.
Now, we introduce the class of Hamiltonians associated with the Dirac operators.
By $\cM_2(\R)$ we denote the set of $2 \times 2$ matrices with real entries.
\begin{definition*}
    $\cG = \cG_{\g}$ is the set of functions $\cH:\R_+ \to \cM^+_2(\R)$ such that
    $$
        \cH' \in L^2(\R_+,\cM_2(\R)),\qq \sup \supp \cH' = \g,\qq \cH(0) = I_2
    $$
    and $\cH$ has the following form
    \[ \label{p1e12}
        \cH = \ma \ga & \gb \\ \gb & \frac{1+\gb^2}{\ga} \am,
    \]
    where $\ga : \R_+ \to \R_+$ and $\gb : \R_+ \to \R$.
\end{definition*}
\begin{remark}
    Let $\cH \in \cG$. Then it follows from (\ref{p1e12}) that
    $$
        \cH^*(x) = \cH(x),\qq \det \cH(x) = 1,\qq \ga(x) > 0,\qq \forall x \in \R_+
    $$
    and $\cH(x)$ is a constant matrix for any $x \geq \g$.
\end{remark}
\begin{theorem} \label{t7}
    The mapping $q \mapsto \cH_q = r^*(\cdot,q)r(\cdot,q)$ from $\cP$ into $\cG$ is a bijection.
    Moreover, if $\cH_q$ has form (\ref{p1e12}) and $q = -q_2 + i q_1$, then we have
    \[ \label{p1e13}
        \begin{aligned}
            q_1 &= -\frac{1}{2}(\gq \cos \gr + \gp \sin \gr),\qq q_2 = \frac{1}{2}(\gp \cos \gr - \gq \sin \gr),\\
            \gp &= \frac{\ga'}{\ga},\qq \gq = \frac{\ga\gb' - \ga'\gb}{\ga},\qq \gr(x) = \int_0^x \gq(\t) d\t,\qq x \in \R_+.
        \end{aligned}
    \]
\end{theorem}
\begin{remark}
    1) One can consider a Hamiltonian of more general type
    $
        \cH = \left( \begin{smallmatrix} \ga & \gb \\ \gb & \gc \end{smallmatrix}\right),
    $
    which satisfies other conditions of $\cG$. Changing the variables in the canonical system,
    one can obtain an equivalent Hamiltonian from $\cG$ (see details in Proposition 1 in \cite{R14}).

    2) Recall that operators $H_D$ with potential $V_q$ for some $q \in \cP$ and $H_{\cH_q}$ are
    unitary equivalent with the unitary transform $U f = r(\cdot,q) f$.

    3) Due to (\ref{p1e13}), if $\gb = 0$, i.e. $\cH$ is a diagonal matrix, then $\gq = \gr = 0$ and
    then $q_1 = 0$ and $q_2 = \frac{\ga'}{2\ga}$. Thus, $H_D$ is
    a supersymmetric Dirac operator and its square is a direct sum of two Schr{\"o}dinger operators with
    singular potentials (see e.g. \cite{T92}).
\end{remark}

It is well-known that for any Hamiltonian of the canonical system there exists a
Hermite-Biehler function $E$ such that $E$ is entire and
$|E(z)| > |E(\ol z)|$ for each $z \in \C_+$ and the associated de Branges space $B(E)$ is given by
$$
    B(E) = \Big\{ F:\C \to \C \, \mid \, \text{$F$ is entire},\, \frac{F}{E},\,
    \frac{F^{\#}}{E} \in \mH^2(\C_+)\, \Big\},
$$
where $F^{\#}(z) = \ol{F(\ol z)}$ and $\mH^2(\C_+)$ is the Hardy space in the upper half-plane.
Moreover, from a de Branges space, one can recover the associated canonical system
(see Theorem 40 in \cite{dB} or Theorem 10, 13 in \cite{R14}).
We introduce a fundamental vector-valued solutions
$\vp = \left( \begin{smallmatrix} \vp_1 \\ \vp_2 \end{smallmatrix} \right)$ and
$\vt = \left( \begin{smallmatrix} \vt_1 \\ \vt_2 \end{smallmatrix} \right)$
of equation (\ref{p1e4})
such that $M = (\vt,\vp)$. If the Hamiltonian can be obtained from Dirac equation
(\ref{p1e4}), then we can construct the associated Hermite-Biehler function as follows
\[ \label{p1e11}
    E(z) = \vp_1(\g,z) - i \vp_2(\g,z),\qq z \in \C,
\]
see e.g. p.~8 in \cite{MM17}. We say that a Hermite-Biehler function is \textit{Dirac-type} if
it has form (\ref{p1e11}), where $\vp$ is a solution of (\ref{p1e4}) for some $q \in \cP$.
Below, we obtain the relation between $\vp(\g,z)$ and $\psi(z)$ (see Proposition \ref{pr1}).
Thus, using this relation and Theorem \ref{t1},
we give the following characterization of the Dirac-type Hermite-Biehler functions.
\begin{corollary} \label{c2}
    A Hermite-Biehler function $E$ is Dirac-type if and only if
    \[ \label{p1e10}
        E(z) = -ie^{-i\g z} \psi_0(z),\qq z \in \C,
    \]
    for some $\psi_0 \in \cJ_0$.
\end{corollary}
\begin{remark}
    Using (\ref{p1e10}), we can apply the results for Jost functions to the Dirac-type
    Hermite-Biehler functions. In particular, due to Corollary \ref{p1c2},
    a Dirac-type Hermite-Biehler function $E$ is uniquely determined by its zeros in $\C_-$
    and the results of Theorem \ref{t5} hold true for the zeros of $E$.
\end{remark}

\section{Preliminary} \label{hl}

\subsection{Notations}
In these section, we recall results about inverse scattering problem
for the Dirac operator on the half-line. We introduce the following
Banach spaces
$$
    \begin{aligned}
        \cB_+ &= L^2(\R_+) \cap L^1(\R_+),\qq
        \| \cdot \|_{\cB_+} = \| \cdot \|_{L^2(\R_+)} + \| \cdot \|_{L^1(\R_+)},\\
        \cB &= L^2(\R) \cap L^1(\R),\qq
        \| \cdot \|_{\cB} = \| \cdot \|_{L^2(\R)} + \| \cdot \|_{L^1(\R)}.\\
    \end{aligned}
$$
We also define the following unital Banach algebras with pointwise multiplication
$$
    \begin{aligned}
        \cA_+ &= \{ \, c + \cF g(x),\, x \in \ol\C_+ \, \mid \, c \in \C,\, g \in \cB_+ \,\},\qq
        \| c + \cF g \|_{\cA_+} = |c| + \| g \|_{\cB_+},\\
        \cA &= \{ \, c + \cF g(x),\, x \in \R \, \mid \, c \in \C,\, g \in \cB \,\},\qq
        \| c + \cF g \|_{\cA} = |c| + \| g \|_{\cB},
    \end{aligned}
$$
where $\cF g(x) = \int_{\R} g(s)e^{2isx}ds$ is the Fourier transform
of $g$. It is well-known that $\cA_+$ and $\cA$ are unital Banach
algebras (see e.g. Chapter 17 in \cite{GRS64}). At last, by
$\cM_2(\C)$ we denote the Banach space of $2 \times 2$ matrices with
complex entries.

\subsection{Jost functions}

We consider Dirac operator $H y = -i \s_3 y' + i \s_3 Q y$, where
potential $Q$ has form (\ref{intro:potential}) and $q \in \cB_+$. If
$q \in \cB_+$, then we introduce the Jost solution $f(x,z)$ of Dirac
equation (\ref{intro:equation}) satisfying the following asymptotic
condition:
\[ \label{p3e1}
    f(x,z) = e^{i z x \s_3} \left( 1 + o(1) \right)\qq \text{as $x \to +\iy$}.
\]
It is well known that for every $z \in \R$ there exists a unique
Jost solution and it has an integral representation in terms of the
transformation operator. We recall these known results,
see p.39 in  \cite{FT07} and  Proposition 3.5 in \cite{FHMP09}.

\begin{lemma} \label{hll1}
    Let $q \in \cB_+$. Then there exists a function $\G:\R_+^2 \to \cM_2(\C)$ such that
    $$
        \G(x,s) = \ma \G_{11} & \G_{12} \\ \G_{21} & \G_{22} \am (x,s),\qq (x,s) \in \R_+^2,
    $$
    and
    \[ \label{hle2}
        f(x,z) = e^{i x z \s_3} + \int_0^{+\iy} \G(x,s) e^{i (2s+x) z \s_3} ds,
        \qq (x,z) \in \R_+ \ts \R.
    \]
    Moreover, the following statements hold true:
    \begin{enumerate}[i)]
        \item For each $n,m = 1,2$, the mapping $x \mapsto \G_{nm}(x, \cdot)$ from $\R_+$ into $\cB_+$
        is continuous and for any $x \in \R_+$, the following estimate holds true:
        \[ \label{hle3}
            \| \G_{nm}(x,\cdot) \|_{\cB_+} \leq e^{\eta(x)}(1 + \zeta(x)) - 1,
        \]
        where
        $$
            \eta(x) = \int_x^{+\iy} |q(s)| ds,\qq
            \zeta(x) = \left( \int_x^{+\iy} |q(s)|^2 ds \right)^{1/2};
        $$
        \item For each $n,m = 1,2$ and each fixed $x \in \R_+$, the mapping
        $q \mapsto \G_{nm}(x, \cdot, q)$ from $\cB_+$ into
        $\cB_+$ is continuous;
        \item For almost all $x \in \R_+$, we have
        \[ \label{p3e15}
            q(x) = -\G_{12}(x,0).
        \]
    \end{enumerate}
\end{lemma}
Recall that the Jost function $\psi(z)$ was defined by (\ref{p2e10}).
We introduce the following class of all Jost functions for
potentials from $\cB_+$.
\begin{definition*}
    $(\cJ_{+}, \rho_{\cJ_{+}})$ is a metric space, where $\cJ_{+}$ is the set of all
    functions $\psi:\ol\C_+ \to \C$ such that
    \[ \label{p2e21}
        \psi(z) = e^{-i\a}  + \int_0^{+\iy} g(s) e^{2 i z s} ds,\qq z \in \ol\C_+,
    \]
    for some $g \in \cB_+$ and $\psi(z) \neq 0$ for any $z \in \overline \C_+$, and
    the metric $\rho_{\cJ_{+}}$ is given by
    \[ \label{p2e22}
        \rho_{\cJ_{+}}(\psi_1,\psi_2) = \|g_1 - g_2\|_{\cB_+},\qq \psi_1,\psi_2 \in \cJ_{+}.
    \]
\end{definition*}
\begin{remark}
    It follows from (\ref{p2e21}), (\ref{p2e22}) that $\cJ_{+} \ss \cA_+$ isometrically. Moreover, each $\psi \in \cJ_{+}$ is
    invertible in $\cA_+$ (see e.g. Lemma 2.9 in \cite{HM16}).
\end{remark}
Using integral representation (\ref{hle2}) and the fact that $\psi
\in \cJ_{+}$ is invertible, we obtain the following corollary.

\begin{corollary} \label{hll2}
    Let $q \in \cB_+$. Then $\psi \in \cJ_{+}$ and $g$ given by (\ref{p2e21}) has the following form:
    \[ \label{p3e12}
        g(s) = e^{-i\a}\G_{11}(0,s) - e^{i\a}\G_{21}(0,s),\qq s \in \R_+.
    \]
    Moreover, there exists a unique $h \in \cB_+$ such that
    \[ \label{p3e10}
        \psi^{-1}(z) = e^{i\a} + \int_0^{+\iy} h(s) e^{2 i z s} ds,\qq z \in \ol\C_+.
    \]
\end{corollary}

\subsection{Direct scattering}
Recall that the scattering matrix $S(z)$ was defined by
(\ref{p2e3}). Now, we introduce the class of all scattering matrices
for potentials from $\cB_+$.
\begin{definition*} \label{hlscat_class}
    $(\cS_{+},\rho_{\cS_{+}})$ is a metric space, where $\cS_{+}$ is the set of all
    continuous functions $S:\R \to \S^1$ such that $W(S) = 0$ and there
    exists
    $F \in \cB$ such that
    \[ \label{hle1}
        S(z) = e^{2i\a}  + \int_{-\iy}^{+\iy} F(s) e^{2 i z s} ds,\qq z \in \R;
    \]
    and the metric $\rho_{\cS_{+}}$ is given by
    $$
        \rho_{\cS_{+}}(S_1,S_2) = \|F_1 - F_2\|_{\cB},\qq
        S_1,S_2 \in \cS.
    $$
\end{definition*}
\begin{remark}
    Note that $\cS \ss \cS_{+} \ss \cA$ isometrically and each $S \in \cS_{+}$ is invertible in $\cA$.
\end{remark}
We need the following lemma.
\begin{lemma} \label{p3l1}
    The mapping $\psi \mapsto S = \ol{\psi} \psi^{-1}$ from $\cJ_{+}$ to $\cS_{+}$ is continuous.
\end{lemma}
\begin{proof}
    Firstly, we show that $S \in \cS_{+}$. Let $\psi \in \cJ_{+}$ and let
    $S(z) = \ol{\psi(z)} \psi^{-1}(z)$, $z \in \R$. Using the definition of $S$, we have
    $|S(z)| = |\ol{\psi(z)}| / | \psi(z)| = 1$, $z \in \R$.
    Let $\psi = e^{-i\a} + \cF g$ for some $g \in \cB_+$.
    Changing variables in the Fourier transform, we get
    $\ol{\psi} = e^{i\a} + \cF r$, where $r(x) = \ol{g(-x)}$,
    $x \in \R$. Since $\psi$ is invertible in $\cB_+ \ss \cA_+$,
    there exists $h \in \cB_+$ such that
    $\psi^{-1} = e^{i\a} + \cF h$. Substituting these representations in the definition of $S$, we get
    \[ \label{p3e13}
        S = e^{2i\a} + \cF F,\qq
        F(s) = e^{-i\a}(h(s) + r(s)) + (r*h)(s),
        \qq s \in \R.
    \]
    Due to $r,h \in \cB$, we have $F \in \cB$. Now, we show that $W(S) = 0$.
    It follows from the definition of $S$ that $W(S) = -2W(\psi)$. We introduce the conformal mapping
    $w \mapsto z(w) = i\frac{1-w}{1+w}$ from $\{|z| \leq 1 \}$ into $\ol\C_+$
    and the function $\phi(w) = \psi(z(w))$, $w \in \{|z| \leq 1 \}$. Thus, $\phi$ is
    analytic on $\{|z| < 1 \}$, continuous up to the boundary $\S^1$, and does not vanish in
    $\{|z| \leq 1 \}$. Moreover, $W(\psi)$ equals the winding number of $\phi(w)$ around zero when
    $w$ runs throughout $\S^1$. Using the Cauchy argument principle, it is easy to see that
    the winding number of $\phi(w)$ around zero, when $w$ runs throughout $|z| = r$, equals zero
    for each $r < 1$. Considering $r \to 1$, we have $W(\psi) = 0$.

    Secondly, we show that the mapping $\psi \mapsto S$ is continuous.
    Since $\psi, \psi^{-1} \in \cA_+ \ss \cA$, and the mapping $x \mapsto x^{-1}$ is continuous
    on the subspace of invertible elements, it follows that $\psi \mapsto \psi^{-1}$ is a
    continuous mapping from $\cA$ to $\cA$. Moreover, the multiplication and the complex
    conjugate are also continuous mappings from $\cA$ to $\cA$. Then we have that
    $\psi \mapsto S = \ol{\psi}\psi^{-1}$ is a continuous mapping from $\cA$ to $\cA$.
    Due to $\cS_{+} \ss \cA$ and $\cJ_{+} \ss \cA$ isometrically, we have that the mapping
    $\psi \mapsto S$ is a continuous mapping from $\cJ_{+}$ to $\cS_{+}$.
\end{proof}

\subsection{Inverse scattering}
We need the following theorem about inverse scattering problem (see e.g. Theorem 3.1 in \cite{HM16}).
\begin{theorem} \label{hlt1}
    The mapping $q \mapsto S(\cdot,q)$ from $\cB_+$ to $\cS_{+}$ is a homeomorphism.
\end{theorem}
It follows from this theorem that a potential is uniquely determined
by a scattering matrix. Moreover, one can recover a potential from a
scattering matrix using the Gelfand-Levitan-Marchenko (GLM)
equation. For any scattering matrix $S \in \cS_{+}$, we introduce the
matrix-valued function
\[ \label{hlF_scat}
    \Omega(x) = \ma 0 & F(-x) \\ \ol{F(-x)} & 0\am,\qq x \in \R,
\]
where $F$ is given by (\ref{hle1}). We also need the following results (see e.g. \cite{HM16}).

\begin{lemma} \label{hll3}
    \begin{enumerate}[i)]
        \item Let $\G(x,s) = \G(x,s,q)$ and $\Omega(s) = \Omega(s,q)$ for
        some $q \in \cB_+$ and for any $x,s \in \R_+$. Then $\G$ and $\Omega$ satisfy the GLM equation
        \[ \label{GLM}
            \G(x,s) + \Omega(x+s) + \int_0^{+\iy} \G(x,t) \Omega(x+t+s) dt = 0
        \]
        for almost all $x,s \in \R_+$.
        \item Let $\Omega$ be given by (\ref{hlF_scat}) for some $S \in \cS_{+}$. Then equation
        (\ref{GLM}) has a unique solution $\G(x,\cdot) \in \cB_+ \otimes \cM_2(\C)$ for any $x \in \R_+$
        and this solution depends continuously on $x \in \R_+$. Moreover, the mapping $s \mapsto \G_{12}(\cdot,s)$ from
        $\R_+$ into $\cB_+$ is continuous.
    \end{enumerate}
\end{lemma}
\begin{remark}
    One can recover a potential $q$ from the scattering matrix $S \in \cS_{+}$ as follows:
    \begin{enumerate}[(i)]
        \item Construct $\Omega$ by $S$ as in (\ref{hlF_scat});
        \item Construct $\G(x,s)$ as a solution of (\ref{GLM});
        \item Recover a potential by using $q(x) = -\G_{12}(x,0)$, $x \in \R_+$.
    \end{enumerate}
\end{remark}

\subsection{Compactly supported potentials}
Now, we show that a potential is compactly supported
if and only if the associated kernel $\G$ is compactly supported.

\begin{lemma} \label{p4l1}
    Let $q \in \cB_+$ and let $\d > 0$. Then $\sup \supp q \leq \d$ if and only if $\G(x,s) = 0$
    for almost all $x, s \in \R_+$ such that $x + s > \d$.
\end{lemma}
\begin{proof}
    Let $\sup \supp q \leq \d$. Then it follows from (\ref{hle3}) that $\| \G_{nm}(x,\cdot) \|_{\cB_+} = 0$
    for each $x > \d$, $n,m = 1,2$. Thus, $\G(x,s) = 0$ for each $x > \d$ and for almost all $s \in \R_+$.
    Substituting this identity in (\ref{GLM}), we get $\Omega(x+s) = 0$ for each $x > \d$ and
    for almost all $s \in \R_+$, i.e. $\Omega(x) = 0$ for almost all $x > \d$.
    Now substituting this identity in (\ref{GLM}), we get $\G(x,s) = 0$ for almost all $x,s \in \R_+$
    such that $x+s > \d$.

    Let $\G(x,s) = 0$ for almost all $x,s \in \R_+$ such that $x+s > \d$. By Lemma \ref{hll3}
    the mapping $s \mapsto \G_{12}(\cdot,s)$ is continuous. Combining these facts, we get
    $\G_{12}(x,0) = 0$ for almost all $x > \d$, which yields, by Lemma
    \ref{hll1}, that $q(x) = 0$ for almost all $x > \d$.
\end{proof}
The support of a potential is also related to the support of the Fourier transform of the
scattering matrix and the Jost function.
\begin{lemma} \label{p3l2}
    Let $q \in \cB_+$ and let $g$ and $F$ be given by (\ref{p2e21}) and (\ref{hle1}).
    Then we have
    \[ \label{p3e11}
        \sup \supp q = -\inf \supp F = \sup \supp g.
    \]
\end{lemma}
\begin{proof}
    Firstly, we show that $\sup \supp q \leq -\inf \supp F$.
    If $\inf \supp F = -\iy$, then the inequality is evident. Let $\inf \supp F = -\d < 0$ for some
    $\d < +\iy$. Due to (\ref{hlF_scat}), we have $\Omega(x) = 0$ for any $x > \d$.
    Substituting this identity in (\ref{GLM}), we get $\G(x,s) = 0$ for almost all $x,s \in \R_+$
    such that $x+t > \d$. Thus, by Lemma \ref{p4l1}, $\sup \supp q \leq \d$.

    Secondly, we show that $-\inf \supp F \leq \sup \supp g$. Let $\sup \supp g = \d$.
    It follows from (\ref{p3e13}) that
    \[ \label{p3e14}
        F(s) = e^{-i\a}(h(s) + r(s)) + (r*h)(s),\qq s \in \R,
    \]
    where $r(s) = \overline{g(-s)}$, $s \in \R$ and $h \in \cB_+$.
    Using $\inf \supp r = -\d$, $\inf \supp h \geq 0$, and well-known property
    $\supp (r * h) \ss \supp r + \supp h$, we get
    $\inf \supp r*h \geq -\d$. Substituting these inequalities in (\ref{p3e14}),
    we have $-\inf \supp F \leq \d$.

    Thirdly, we show that $\sup \supp g \leq \sup \supp q$. Let $\sup \supp q = \d$. Then, by Lemma
    \ref{p4l1}, $\G(0,s) = 0$ for almost all $s > \d$. Using (\ref{p3e12}), we get $\sup \supp g \leq \d$.

    Combining these three inequalities, we obtain (\ref{p3e11}).
\end{proof}

\section{Proof of the main theorems} \label{p4}

\begin{proof}[\bf Proof of Theorem \ref{t1}]
    i) Since $\cP \ss \cB_+$ and the metrics $\rho_{\cP}$ and $\| \cdot \|_{\cB_+}$ are equivalent on
    $\cP$, it follows from Theorem \ref{hlt1} that the mapping $q \mapsto S(\cdot,q)$ from $\cP$ to
    $\cS_{+}$ is continuous. Moreover, by Lemma \ref{p3l2}, this mapping is a bijection between $\cP$ and $\cS$.
    Due to the fact that $\cS \ss \cS_{+}$ isometrically, we have that the inverse mapping is also
    continuous and then it is a homeomorphism between $\cP$ and $\cS$.

    ii) We show that the mapping $q \mapsto \psi(\cdot,q)$ is a bijection between $\cP$ and
    $\cJ$. Let $q \in \cP \ss \cB_+$. Then, by Lemma \ref{hll2}, we have $\psi \in \cJ_{+}$. Due to
    $\sup \supp q = \g$, it follows from Lemma \ref{p3l2} that $\psi \in \cJ$.

    Let $\psi \in \cJ \ss \cJ_{+}$. It follows from Lemma \ref{p3l1} that
    $S = \ol{\psi} \psi^{-1} \in \cS_{+}$. Moreover, using Lemma \ref{p3l2} and $\psi \in \cJ$,
    we get $S \in \cS$. Thus, by i), there exists a unique $q \in \cP$ such that
    $S(\cdot) = S(\cdot,q)$. It is easy to see that $\psi(\cdot) = \psi(\cdot,q)$, since they are
    entire functions and their zeros coincide.

    Now, we show that the mapping $q \mapsto \psi(\cdot,q)$ from $\cP$ into $\cJ$ and its inverse are
    continuous. Let $q_1,q_2 \in \cP$.
    Then using (\ref{p2e7}), (\ref{p3e12}), and the fact that the norms
    $\|\cdot\|_{L^2(0,\g)}$ and $\| \cdot \|_{\cB_+}$ are equivalent on $\cP$, we get
    \[ \label{p4e4}
        \begin{aligned}
            \rho_{\cJ}(\psi(\cdot,q_1), \psi(\cdot,q_2)) &=
            \| e^{-i\a}(\G_{11}^{(1)}(0,\cdot)-\G_{11}^{(2)}(0,\cdot)) -
            e^{i\a}(\G_{21}^{(1)}(0,\cdot)-\G_{21}^{(2)}(0,\cdot)) \|_{L^2(0,\g)}\\
            &\leq \| \G_{11}^{(1)}(0,\cdot)-\G_{11}^{(2)}(0,\cdot)\|_{\cB_+} +
            \| \G_{21}^{(1)}(0,\cdot)-\G_{21}^{(2)}(0,\cdot) \|_{\cB_+},
        \end{aligned}
    \]
    where $\G_{nm}^{(j)}(0,\cdot) = \G_{nm}(0,\cdot,q_j)$ for any $n,m,j = 1,2$.
    Due to Lemma \ref{hll1}, the mappings $q \mapsto \G_{11}(0,\cdot,q)$ and
    $q \mapsto \G_{12}(0,\cdot,q)$ from $\cB_+$ into $\cB_+$ are continuous. The
    embedding $\cP \ss \cB_+$ is also continuous.
    Thus, it follows from (\ref{p4e4}) that $\rho_{\cJ}(\psi(\cdot,q_1), \psi(\cdot,q_2)) \to 0$
    as $\rho_{\cP}(q_1,q_2) \to 0$.

    Let $\rho_{\cJ}(\psi(\cdot,q_1), \psi(\cdot,q_2)) \to 0$, where $q_1,q_2 \in \cP$.
    Then it follows from Lemma \ref{p3l1} that $\rho_{\cS}(S(\cdot,q_1), S(\cdot,q_2)) \to 0$.
    Since the mapping $q \mapsto S(\cdot,q)$ from $\cP$ to $\cS$ is a homeomorphism,
    we get $\rho_{\cP}(q_1,q_2) \to 0$.

    iii) By i) and ii), the mappings $\psi(\cdot,q) \mapsto q$ from $\cJ$ into $\cP$ and
    $q \mapsto S(\cdot,q) = \ol{\psi(\cdot,q)}\psi^{-1}(\cdot,q)$ from $\cP$ into $\cS$ are
    homeomorphisms. Combining these mappings, we get that the mapping
    $\psi \mapsto S = \ol{\psi} \psi^{-1}$ from $\cJ$ into $\cS$ is a homeomorphism. Moreover, since
    $\psi \in \cJ$ is entire, it follows that $S$ admits a meromorphic continuation from $\R$ onto $\C$.
\end{proof}

\begin{proof}[\bf Proof of Corollary \ref{p1c2}]
    It follows from Theorem \ref{t1} that $q \in \cP$ is uniquely determined by its Jost function
    $\cJ$. Each function from $\cJ$ is entire and then it is uniquely determined by its zeros.
    Since the resonances are zeros of the Jost function, we get that $q \in \cP$ is uniquely
    determined by its resonances. By the Paley-Wiener theorem (see e.g. p.30 in \cite{Koo98}),
    the Jost function belongs to $\cE_{Cart}$ and then it satisfies (\ref{p2e13}-13).
\end{proof}

\begin{proof}[\bf Proof of Theorem \ref{t5}]
    It follows from Theorem \ref{t1} that $\psi(\cdot,q) \in \cJ$ and then there exists $g \in \cP$
    such that $\psi(\cdot,q) = e^{-i\a} + \cF g$. It is well-known that the set of smooth
    compactly supported functions $C_o^{\iy}(0,\g)$ is dense in $L^2(0,\g)$. Thus, for any $\ve > 0$,
    there exists $g_1 \in C_o^{\iy}(0,\g)$ such that $g = g_1 + g_2$ and $\| g_2 \|_{L^2(0,\g)} < \ve |\g|^{-1/2}$.
    Let $\psi(z) = 0$ for some $z \in \C_-$. Then we have $\cF g(z) = -e^{-i\a}$.
    Estimating the left-hand side of this identity, we get
    \[ \label{p4e18}
        |\cF g_1(z)| + |\cF g_2(z)| \geq 1.
    \]
    Due to $g_1 \in C_o^{\iy}(0,\g)$, we obtain
    \[ \label{p4e19}
        \begin{aligned}
            |\cF g_1(z)| &\leq \left| \int_0^{\g} g_1(s) e^{2izs} ds \right| =
            \left| \frac{-1}{2iz}\int_0^{\g} g'_1(s) e^{2izs} ds \right| \\
            &\leq \frac{1}{2|z|} \int_0^{\g} |g'_1(s)| e^{-2s \Im z} ds \leq
            \frac{e^{-2\g \Im z}}{2|z|} \|g'_1 \|_{L^1(0,\g)} = C e^{-2 \g \Im z}\frac{1}{|z|}.
        \end{aligned}
    \]
Due to $\| g_2 \|_{L^1(0,\g)} \leq \sqrt{\g} \| g_2 \|_{L^2(0,\g)} =
\ve$, we have
    \[ \label{p4e20}
        |\cF g_2(z)| \leq \int_0^{\g} |g_2(s)| e^{-2s\Im z} ds \leq
        e^{-2 \g \Im z} \| g_2 \|_{L^1(0,\g)} = \ve e^{-2 \g \Im z}.
    \]
    Substituting (\ref{p4e19}), (\ref{p4e20}) in (\ref{p4e18}), we get
    $$
        e^{-2\g \Im z} \left(\ve + \frac{C}{|z|} \right) \geq 1,
    $$
    which yields (\ref{p1e1}). Now we consider (\ref{p1e1}) for $|z_n| \to \iy$.
    For $\ve > 0$ and $C \geq 0$ fixed, we~get
    $$
        2\g \Im z_n \leq \ln(\ve) + O(|z|^{-1}).
    $$
    Thus, there are finitely many resonances such that $\Im z_n > \ln(\ve)$.
    Since it holds for any $\ve > 0$, we complete the proof of the lemma.
\end{proof}

\begin{proof}[\bf Proof of Theorem \ref{t3}]
    Let $q \in \cP$. Then, by Theorem \ref{t1}, $S(\cdot) = S(\cdot,q) \in \cS$, which yields that
    $|S(z)| = 1$, $z \in \R$, and $W(S) = 0$. Then there exist a real-valued function
    $\phi_{sc} \in L^{\iy}(\R)$ such that (\ref{p2e9}) holds. By Theorem \ref{t1},
    $S$ is a meromorphic function without poles on the real line. Thus, we get $\phi_{sc} \in C^{\iy}(\R)$.

    By (\ref{p2e4}), $S(\cdot) - e^{2 i \a}$ is a Fourier transform of a function from $L^1(\R)$.
    Then the application of the Riemann-Lebesgue lemma (see e.g. Theorem IX.7 in \cite{RS80}) yields
    $S(z) - e^{2 i \a} \to 0$ as $z \to \pm \iy$ and then $\phi_{sc}(z) \to -\a + 2\pi n_{\pm}$ as
    $z \to \pm \iy$ for some $n_+, n_- \in \Z$. Since $W(S) = 0$, it follows that
    $n_+ = n_-$. Thus, we can choose $\phi_{sc}$ such that $\phi_{sc}(z) \to -\a$ as $z \to \pm \iy$.

    It follows from (\ref{p2e4}) and the Plancherel theorem (see e.g. Theorem IX.6 in \cite{RS80})
    that $S(\cdot) - e^{2 i \a} \in L^2(\R)$. Thus, there exists $g \in L^2(\R)$ such that
    \[ \label{p4e8}
        1-e^{-2 i (\phi_{sc}(z)+\a)} = g(z),\qq z \in \R.
    \]
    Above, we show that $\phi_{sc}(z) \to -\a$ as $z \to \pm \iy$.
    Then using the Taylor series for the exponential function in (\ref{p4e8}), we get
    $$
        \phi_{sc}(z)+\a = O(g(z))\qq \text{as $z \to \pm \iy$}.
    $$
    Since $g \in L^2(\R)$ and $\phi_{sc}(\cdot)+\a \in L^{\iy}(\R)$, it follows that
    $\phi_{sc}(\cdot)+\a \in L^2(\R)$.

    Recall that $\psi(z) \neq 0$ for any $z \in \R$. Thus, we have
    $$
        \int_0^{z} \frac{\psi'(s)}{\psi(s)} ds = \ln \psi(z) - \ln \psi(0),\qq z \in \R,
    $$
    which yields
    \[ \label{s4e1}
        \arg \psi(z) = \arg \psi(0) + \int_0^{z} \Im \frac{\psi'(s)}{\psi(s)} ds,\qq z \in \R.
    \]
    It follows from (\ref{p2e3}) that $\phi_{sc}(z) = \arg \psi(z)$, $z \in \R$. Due to (\ref{s4e1}),
    we get $\phi'_{sc}(z) = \Im \frac{\psi'(z)}{\psi(z)}$, $z \in \R$. It follows from Corollary
    \ref{p1c2} that $\psi \in \cE_{Cart}$. Using the Hadamard factorization (\ref{p2e13}) for $\psi$,
    we obtain
    $$
        \Im \frac{\psi'(z)}{\psi(z)} = \g + \sum_{n = 1}^{\iy} \frac{\Im z_n}{|z - z_n|^2},\qq
        z \in \C \sm \{z_n\}_{n = 1}^{\iy},
    $$
    where the series converges absolutely and uniformly on compact subsets of
    $\C \sm \{z_n\}_{n = 1}^{\iy}$, since $\Im z_n < 0$ for each $n \geq 1$ and
    (\ref{p2e14}) holds. Considering (\ref{s4e1}) as $z \to \pm \iy$ and using
    $\phi_{sc}(z) \to -\a$ as $z \to \pm \iy$, we obtain
    $$
        \phi_{sc}(0) = -\a - \lim_{z \to +\iy} \int_{0}^{z} \phi'_{sc}(s) ds =
            -\a + \lim_{z \to -\iy} \int_{z}^{0} \phi'_{sc}(s) ds.
    $$
\end{proof}

\begin{proof}[\bf Proof of Theorem \ref{t4}]
    For simplicity, we consider the case when $N = 1$. Let
    $\psi = \psi(\cdot,q)$ for some $q \in \cP$ and let $\psi(z_0) = 0$ for some $z_0 \in \C_-$
    and $z_1 \in \C_-$. We introduce
    $$
        B(z) = \frac{z-z_1}{z-z_0},\qq z \in \C \sm \{z_0\}.
    $$
    Firstly, we show that $\psi_1 = B\psi \in \cJ$.
    By Corollary \ref{p1c2}, $\psi$ is entire. Due to $\psi(z_0) = 0$,
    it follows that $\psi_1$ is entire and $\psi_1(z) \neq 0$ for any
    $z \in \ol\C_+$. We show that (\ref{p2e1}) holds for $\psi_1$.
    We consider $h(z) = \psi_1(z) - e^{-i \a}$, $z \in \C$, which is an entire function,
    since $\psi_1$ is entire. Using the definition of $\t_{\pm}$, we have
    $$
        \t_{\pm}(h) = \t_{\pm}(\psi_1) = \t_{\pm}(B) + \t_{\pm}(\psi) = \t_{\pm}(\psi),
    $$
    where $\t_{\pm}(B) = 0$, since $B(z) \to 1$ as $|z| \to \iy$. Due to $\psi \in \cE_{Cart}$, we get
    \[ \label{p4e5}
        \t_+(h) = \t_+(\psi) = 0,\qq \t_-(h) = \t_-(\psi) = 2\g.
    \]
    Now, we show that $h \in L^2(\R)$. We have
    \[ \label{p4e14}
        h(z) = e^{-i \a}(B(z) - 1) + B(z) (\psi(z) - e^{-i \a}),\qq z \in \C.
    \]
    By direct calculation, we get
    \[ \label{p4e7}
        \| 1-B(\cdot)\|_{L^{\iy}(\R)} = \frac{|z_0 - z_1|}{|\Im z_0|},\qq
        \| 1-B(\cdot)\|_{L^2(\R)} = |z_0 - z_1|\left|\frac{\pi}{\Im z_0}\right|^{1/2}.
    \]
    Moreover, it follows from (\ref{p2e1}) and the Plancherel theorem (see e.g. Theorem IX.6 in \cite{RS80})
    that $\psi(\cdot) - e^{-i \a} \in L^2(\R)$. Using these facts, we get from (\ref{p4e14}) that
    $h \in L^2(\R)$. Recall that $h$ is entire and (\ref{p4e5}) holds. Thus, it follows from
    the Paley-Wiener theorem (see e.g. p.30 in \cite{Koo98}) that $h$ has the following form
    $$
        h(z) = \int_0^{\g} r(s) e^{2 i z s} ds,\qq z \in \C,
    $$
    for some $r \in \cP$, which yields that $\psi_1 \in \cJ$. Thus, it follows from Theorem
    \ref{t1} that there exists a unique $q_o \in \cP$ such that $\psi_1(\cdot) = \psi(\cdot,q_o)$.

    Secondly, we show that $\rho_{\cJ}(\psi_1,\psi) \to 0$ as $z_1 \to z_0$.
    Let $\psi = e^{-i\a}+ \cF g$ and $\psi_1 = e^{-i\a}+ \cF r$ for some $g,r \in \cP$.
    Then using (\ref{p2e7}) and the Plancherel theorem, we get
    \[ \label{p4e6}
        \begin{aligned}
            \rho_{\cJ}(\psi_1,\psi) &= \| r - g \|_{L^2(0,\g)} =
            \|\psi_1 - \psi\|_{L^2(\R)} = \|B\psi - \psi\|_{L^2(\R)}\\
            &= \|(\psi - e^{-i\a})(1-B) + e^{-i\a}(1-B)\|_{L^2(\R)}\\
            &\leq \| 1-B\|_{L^2(\R)} +
            \| 1-B\|_{L^{\iy}(\R)} \|\psi - e^{-i\a}\|_{L^2(\R)}.
        \end{aligned}
    \]
    Substituting (\ref{p4e7}) in (\ref{p4e6}), we get
    $$
        \rho_{\cJ}(\psi_1,\psi) \leq |z_0 - z_1|\left|\frac{\pi}{\Im z_0}\right|^{1/2}
        \left( 1 + \frac{\|\psi - e^{-i\a}\|_{L^2(\R)}}{|\pi \Im z_0|^{1/2}} \right),
    $$
    which yields $\rho_{\cJ}(\psi_1,\psi) \to 0$ as $z_1 \to z_0$. By Theorem \ref{t1}, the mappings
    $q \mapsto \psi(\cdot,q)$ and $q \mapsto S(\cdot,q)$ are homeomorphisms and then we have
    $\rho_{\cP}(q_o,q) \to 0$ and $\rho_{\cS}(S(\cdot,q_o),S(\cdot,q)) \to 0$ as $z_1 \to z_0$.
\end{proof}

\begin{proof}[\bf Proof of Theorem \ref{t6}]
    i) Let $\psi(\cdot) = \psi(\cdot,q)$ for some $q \in \cP$ and let $k \in \R$. We show that
    $\psi(\cdot+k) \in \cJ$.
    Since $\psi(z) \neq 0$, $z \in \ol{ \C_+}$, it follows that $\psi(z+k) \neq 0$ for each
    $z \in \ol\C_+$. Using (\ref{p2e1}), we have
    $$
        \psi(z+k) = e^{-i\a} + \int_0^{\g} g(s) e^{2 i k s} e^{2 i z s} ds,\qq z \in \C,
    $$
    where $g(s) e^{2 i k s} \in \cP$. Thus, $\psi(\cdot+k) \in \cJ$ and it follows from Theorem
    \ref{t1} that there exists a unique $q_k \in \cP$ such that $\psi(\cdot+k) = \psi(\cdot,q_k)$.
    Moreover, Theorem \ref{t1} gives that $S(\cdot + k) = S(\cdot,q_k) \in \cS$.
    Now, we recover $q_k$. Using (\ref{p2e4}), we get
    $$
        S(z+k) = e^{2i\a} + \int_{-\g}^{+\iy} F(s) e^{2i(z+k)s} ds =
        e^{2i\a} + \int_{-\g}^{+\iy} F_{k}(s) e^{2izs} ds,\qq z \in \C.
    $$
    We introduce
    $$
        \Omega(x) = \ma 0 & F(-x) \\ \ol{F(-x)} & 0 \am,\qq \Omega_{k}(x) = \ma 0 & F_{k}(-x) \\ \ol{F_{k}(-x)} & 0 \am,\qq x \in \R,
    $$
    where
    \[ \label{p4e16}
        \Omega_{k}(x) = \Omega(x) e^{2 i k x \s_3} = e^{-2 i k x \s_3} \Omega(x),\qq x \in \R.
    \]
    It follows from (\ref{p3e15}) and Lemma \ref{hll3} that there exist a unique solution
    $\G(x,s)$ of the GLM equation
    \[ \label{p4e15}
        \G(x,s) + \Omega(x+s) + \int_0^{+\iy} \G(x,t) \Omega(x+t+s) dt = 0
    \]
    such that $\G_{12}(x,0) = -q(x)$ for almost all $x \in \R_+$. We introduce
    \[ \label{p4e13}
        \G_{k}(x,s) = e^{i x k \s_3} \G(x,s) e^{-i(2s+x)k \s_3},\qq x,s \in \R_+.
    \]
    Substituting (\ref{p4e13}) in (\ref{p4e15}) and using (\ref{p4e16}), we get for almost all $x,s \in \R_+$
    $$
        \G_{k}(x,s) + \Omega_{k}(x+s) + \int_0^{+\iy} \G_{k}(x,t) \Omega_{k}(x+t+s) dt = 0.
    $$
    Thus, it follows from Lemma \ref{hll3} that $\G_{k}(x,s) = \G(x,s,q_k)$
    and then, due to (\ref{p3e15}) and (\ref{p4e13}), we get for almost all $x \in \R_+$
    $$
        q_{k}(x) = -(\G_{k})_{12}(x,0) = -e^{2 i k x }\G_{12}(x,0) = e^{2 i k x } q(x).
    $$

    ii) Let $\psi(\cdot) = \psi(\cdot,q)$ for some $q \in \cP$ and let
    $\psi_1(z) = e^{-2i\a}\ol{\psi(-\ol{z})}$, $z \in \C$. We show that $\psi_1 \in \cJ$.
    Since $\psi(z) \neq 0$, $z \in \ol{ \C_+}$, it follows that $\psi_1(z) \neq 0$ for each
    $z \in \ol\C_+$. Using (\ref{p2e1}), we have
    $$
        \psi_1(z) = e^{-i\a} + \int_0^{\g} e^{-2i\a}\ol{g(s)} e^{2 i z s} ds,\qq z \in \C,
    $$
    where $e^{-2 i \a}\ol{g(s)} \in \cP$. Thus, $\psi_1 \in \cJ$ and it follows from Theorem
    \ref{t1} that there exists a unique $q_o \in \cP$ such that $\psi_1(\cdot) = \psi(\cdot,q_o)$.
    Moreover, Theorem \ref{t1} gives that $S_1 = \ol{\psi_1(\cdot)} \psi_1^{-1}(\cdot)
    = S(\cdot,q_o) \in \cS$. Now, we recover $q_o$.
    Using the definition of $\psi_1$, we get
    \[ \label{p5e7}
        S_1(z) = \frac{e^{2i\a}\psi(-\ol{z})}{e^{-2i\a}\ol{\psi(-\ol{z})}} = e^{4i\a} \ol{S(-\ol{z})},\qq z \in \R,
    \]
    where $S(\cdot) = S(\cdot,q) \in \cS$. Substituting (\ref{p2e4}) in (\ref{p5e7}), we obtain
    $$
        S_1(z) = e^{2i\a} + \int_{-\g}^{+\iy} e^{4i\a}\ol{F(s)} e^{2izs} ds =
        e^{2i\a} + \int_{-\g}^{+\iy} F_{o}(s) e^{2izs} ds,\qq z \in \R.
    $$
    We introduce
    $$
        \Omega(x) = \ma 0 & F(-x) \\ \ol{F(-x)} & 0 \am,\qq \Omega_{o}(x) = \ma 0 & F_{o}(-x) \\ \ol{F_{o}(-x)} & 0 \am,\qq x \in \R,
    $$
    where
    \[ \label{p4e23}
        \Omega_{o} = U \Omega(x) U^{-1},\qq U = e^{2i\a \s_3} \s_1,\qq
        \s_1 = \ma 0 & 1 \\ 1 & 0 \am.
    \]
    As above, there exist a unique solution $\G(x,s)$ of the GLM equation
    \[ \label{p4e22}
        \G(x,s) + \Omega(x+s) + \int_0^{+\iy} \G(x,t) \Omega(x+t+s) dt = 0
    \]
    such that $\G_{12}(x,0) = -q(x)$ for almost all $x \in \R_+$. We introduce
    \[ \label{p4e21}
        \G_{o}(x,s) = U \G(x,s) U^{-1},\qq x,s \in \R_+.
    \]
    Substituting (\ref{p4e21}) in (\ref{p4e22}) and using (\ref{p4e23}), we get for almost all $x,s \in \R_+$
    $$
        \G_{o}(x,s) + \Omega_{o}(x+s) + \int_0^{+\iy} \G_{o}(x,t) \Omega_{o}(x+t+s) dt = 0.
    $$
    Thus, it follows from Lemma \ref{hll3} that $\G_{o}(x,s) = \G(x,s,q_o)$
    and then, due to (\ref{p3e15}) and (\ref{p4e13}), we get for almost all $x \in \R_+$
    $$
        q_{o}(x) = -(\G_{o})_{12}(x,0) = -e^{4 i \a }\ol{\G_{12}(x,0)} = e^{4 i \a } \ol{q(x)}.
    $$
\end{proof}

\section{Canonical systems}
\begin{proof}[\bf Proof of Theorem \ref{t7}]
    Firstly, we prove that the mapping $q \mapsto \cH_q = r^*(\cdot,q) r(\cdot,q)$ is an injection.
    Let
    \[ \label{p6e2}
        r_1^* r_1 = r_2^* r_2,
    \]
    where $r_i = r(\cdot,q^i)$, $q^i \in \cP$, for $i = 1,2$. Since $\det r_i = 1$, it follows
    that there exists $r_i^{-1}$. Let $U = r_1 r_2^{-1}$. Using (\ref{p6e2}), we get
    $$
        U^* U = (r_1 r_2^{-1})^* r_1 r_2^{-1} = I_2,
    $$
    which yields that $U(x)$ is a unitary matrix for any $x \in \R_+$.
    Moreover, since $r_i$ is differentiable, $r_i(0) = I_2$, $r_i(x)$ has real entries, and
    $\det r_i(x) = 1$ for any $x \in \R_+$, it follows that
    $U$ is differentiable, $U(0) = I_2$, $U(x)$ has real entries, and $\det U(x) = 1$ for any $x \in \R_+$.
    Moreover, it follows from these properties that $U$ has the following form
    \[ \label{p6e3}
        U = \ma a & b \\ -b & a\am.
    \]
    Using $r_i' = JV_{q^i} r_i$ and $U = r_1 r_2^{-1}$, we get the following identity
    $$
        JV_{q^1} = U'U^{*} + U J V_{q^2} U^{*}.
    $$
    It is easy to see that the matrices $JV_{q^1}$ and $U J V_{q^2} U^{*}$ are self-adjoint and
    their traces equal zero. This implies that the matrix $U'U^{*}$ is self-adjoint and its trace
    equals zero. Using (\ref{p6e3}), we obtain
    \[ \label{p6e4}
        U'U^{*} = \ma aa' + bb' & ab' - a'b \\ a'b - ab' & aa' + bb'\am.
    \]
    Since $U'U^{*}$ is self-adjoint and its trace equals zero, it follows from (\ref{p6e4}) that
    $U'U^{*} = 0$, which yields $U' = 0$. Using $U(0) = I_2$, we get $U = I_2$ and then $r_1 = r_2$.
    Since $q$ is uniquely determined by $r(\cdot,q)$, it follows that $q^1 = q^2$.

    Secondly, we prove that the mapping $q \mapsto \cH_q = r^*(\cdot,q) r(\cdot,q)$ from $\cP$ into
    $\cG$ is a surjection. Let $\cH \in \cG$ have the following form
    $$
        \cH = \ma \ga & \gb \\ \gb & \frac{1+\gb^2}{\ga} \am.
    $$
    We introduce the following functions on $[0,\iy)$:
    $$
        \gp = \frac{\ga'}{\ga},\qq \gq = \frac{\ga\gb' - \ga'\gb}{\ga},\qq \gr(x) = \int_0^x \gq(\t) d\t,\qq x \in \R_+.
    $$
We define a  matrix-valued function $r$ by
    \[ \label{p6e5}
        r = UR,\qq U = \ma c & -s \\ s & c\am,\qq
        R = \ma \sqrt{\ga} & \frac{\gb}{\sqrt{\ga}} \\ 0 & \frac{1}{\sqrt{\ga}} \am,
    \]
    where $\sqrt{\ga} > 0$ and $c = \cos\left(\gr/2\right)$, $s = \sin\left(\gr/2\right)$.
    Due to $\cH(0) = I_2$, we deduce that $\ga(0) = 1$ and $\gb(0) = 0$
    and thus $r(0) = I_2$. Then, since $\cH$ is differentiable, we have that $r$ is differentiable.
    Moreover, $U$ is an unitary matrix and $\cH = r^* r = R^* R$.
    Now, we show that $r$ is a solution of the Dirac equation $Jr' + V r = 0$ with a potential $V = V_q$
    for some $q \in \cP$.
    Extracting $V$ from this equation and using (\ref{p6e5}), we get
    \[ \label{p6e7}
        V = -Jr'r^{-1} = -J(U'R + UR')R^{-1}U^* = -JU(U^*U' + R'R^{-1})U^{*}.
    \]
    Differentiating $U$ and $R$, we obtain
    \[ \label{p6e6}
        U' = \frac{1}{2}\ma -\gq s & -\gq c \\ \gq c & -\gq s\am,\qq
        R' = \ma \frac{\ga'}{2\sqrt{\ga}} & \frac{\gb'}{\sqrt{\ga}} - \frac{\gb\ga'}{2\ga^{3/2}} \\
        0 & -\frac{\ga'}{2\ga^{3/2}} \am.
    \]
    Using (\ref{p6e5}) and (\ref{p6e6}), we get
    \[ \label{p6e8}
        \begin{aligned}
            U^*U' &= \frac{1}{2} \ma c & s \\ -s & c \am
            \ma -\gq s & -\gq c \\ \gq c & -\gq s \am = \frac{1}{2} \ma 0 & -\gq \\ \gq & 0 \am,\\
            R'R^{-1} &= \ma \frac{\ga'}{2\sqrt{\ga}} & \frac{\gb'}{\sqrt{\ga}} - \frac{\gb\ga'}{2\ga^{3/2}} \\
            0 & -\frac{\ga'}{2\ga^{3/2}} \am
            \ma \frac{1}{\sqrt{\ga}} & -\frac{\gb}{\sqrt{\ga}} \\ 0 & \sqrt{\ga} \am =
            \frac{1}{2} \ma \gp & 2\gq \\ 0 & -\gp \am.
        \end{aligned}
    \]
    Substituting (\ref{p6e8}) in (\ref{p6e7}), we obtain
    \[ \label{p6e9}
        \begin{aligned}
            V &= \frac{1}{2} \ma 0 & -1 \\ 1 & 0\am \ma c & -s \\ s & c\am
            \ma \gp & \gq \\ \gq & -\gp \am \ma c & s \\ -s & c\am \\
            &= \frac{1}{2} \ma -(c^2 - s^2)\gq - 2cs \gp & (c^2 - s^2)\gp - 2cs \gq \\
            (c^2 - s^2)\gp - 2cs \gq & (c^2 - s^2) \gq + 2cs \gp \am.
        \end{aligned}
    \]
    Recall that $c = \cos\left(\gr/2\right)$ and $s = \sin\left(\gr/2\right)$ and then
    $$
        c^2 - s^2 = \cos \gr,\qq 2cs = \sin \gr.
    $$
    Substituting these identities in (\ref{p6e9}) we get (\ref{p1e13}). Recall that
    $\cH' \in L^2(\R_+,\cM_2(\R))$. This implies that $\gp, \gq \in L^2(\R_+)$ and then
    $q \in L^2(\R_+)$, where $q = -q_2 + iq_1$. Recall also that
    $\sup \supp \cH' = \g$. It follows that
    $$
        \max (\sup \supp \gp, \sup \supp \gq) = \g
    $$
    which yields $\sup \supp q = \g$. Thus, we have $q \in \cP$, $r = r(\cdot,q)$ and $\cH = \cH_q$.
\end{proof}

Recall that $\vp$ and $\vt$ are fundamental vector-valued solutions of (\ref{p1e4}).
We introduce
\[ \label{p1e8}
    u(x,z,\a) = \vp(x,z) \cos \a - \vt(x,z) \sin \a,\qq (x,z,\a) \in \R_+ \ts \C \ts [0,\pi).
\]
Then $u = \left( \begin{smallmatrix} u_1 \\ u_2 \end{smallmatrix} \right)$ is a solution of
(\ref{p1e4}) satisfying the initially condition
\[ \label{p1e6}
    u_1(0,z,\a) \cos \a + u_2(0,z,\a) \sin \a = 0.
\]
Now, we show how this solution is
associated with the Jost function. Recall that $\psi_{\a}(z)$ is the Jost function
given by (\ref{p2e10}) for any $\a \in [0,\pi)$.
\begin{proposition} \label{pr1}
    Let $V_q$ be given by (\ref{p1e5}) for some $q \in \cP$ and let $u$ be given by (\ref{p1e8}).
    Then for any $\a \in [0,\pi)$, we have
    $$
        \psi_{\a}(z) = e^{i\g z}(u_2(\g,z,\a) + i u_1(\g,z,\a)),\qq z \in \C.
    $$
    In particular,
    \[ \label{p1e9}
        \begin{aligned}
            \psi_0(z) &= e^{i\g z}(\vp_2(\g,z) + i \vp_1(\g,z))\\
            \psi_{\pi/2}(z) &= -e^{i\g z}(\vt_2(\g,z) + i \vt_1(\g,z))
        \end{aligned},\qq z \in \C.
    \]
\end{proposition}
\begin{remark}
    Recall that if $\a = 0$, then (\ref{intro:bc}) and (\ref{p1e6}) are the Dirichlet boundary conditions
    and if $\a = \frac{\pi}{2}$, then (\ref{intro:bc}) and (\ref{p1e6}) are the Neumann boundary conditions.
\end{remark}
\begin{proof}[\bf Proof of Proposition \ref{pr1}]
    Since $Tf(x,z)$ is a solution of (\ref{p1e4}), it follows that
    $$
        M(x,z) = T f(x,z) f^{-1}(0,z) T^{-1},\qq (x,z) \in \R_+ \ts \C.
    $$
    Due to $f(\g,z) = e^{i\g z \s_3}$, we get
    $$
        f(0,z) = T^{-1} M^{-1}(\g,z) T e^{i\g z \s_3},\qq z \in \C.
    $$
    Using $T^{-1} = \frac{1}{2}\left( \begin{smallmatrix} -i & 1 \\ i & 1 \end{smallmatrix}\right)$
    and $M^{-1}(\g,z) = \left( \begin{smallmatrix} \vp_2 & -\vp_1 \\ -\vt_2 & \vt_1 \end{smallmatrix}\right)(\g,z)$,
    we obtain
    \[ \label{p6e1}
        f(0,z) = \frac{1}{2}\ma e^{i\g z} (\vt_1 - i \vt_2 + \vp_2 + i \vp_1) &
        e^{-i\g z} (\vt_1 + i \vt_2 - \vp_2 + i \vp_1) \\
        e^{i\g z} (\vt_1 - i \vt_2 - \vp_2 - i \vp_1) &
        e^{-i\g z} (\vt_1 + i \vt_2 + \vp_2 - i \vp_1)\am,\qq z \in \C,
    \]
where $\vp_i = \vp_i(\g,z)$, $\vt_i = \vt_i(\g,z)$, $i =1,2$ for
shortness. Substituting (\ref{p6e1}) in  (\ref{p2e10}) and using
(\ref{p1e8}), we get
    $$
        \begin{aligned}
            \psi_{\a}(z) &= \frac{1}{2} e^{i \g z}( e^{-i\a}(\vt_1 - i \vt_2 + \vp_2 + i \vp_1) -
            e^{i\a}(\vt_1 - i \vt_2 - \vp_2 - i \vp_1))\\
            &= \frac{1}{2} e^{i \g z}(\vt_1 (e^{-i\a} - e^{i\a}) -i \vt_2 (e^{-i\a} - e^{i\a}) +
            \vp_2 (e^{-i\a} + e^{i\a}) + i\vp_1 (e^{-i\a} + e^{i\a}))\\
            &= e^{i \g z}(-i\vt_1 \sin \a - \vt_2 \sin \a +
            \vp_2 \cos \a + i \vp_1 \cos \a)\\
            &= e^{i \g z}(u_2(\g,z,\a) + i u_1(\g,z,\a)).
        \end{aligned}
    $$
    In particular, if $\a = 0$, then $u(\g,z,0) = \vp(\g,z)$ and if $\a = \frac{\pi}{2}$,
    then $u(\g,z,\pi/2) = -\vt(\g,z)$, which yields (\ref{p1e9}).
\end{proof}

\begin{proof}[\bf Proof of Corollary \ref{c2}]
    Let $E$ be given by (\ref{p1e11}), where $\vp$ is a solution of (\ref{p1e4})
    for some $q \in \cP$. Then, by Proposition \ref{pr1}, $E(z) = -i e^{-i\g z}\psi_0(z,q)$, $z \in \C$ and
    it follows from Theorem \ref{t1} that $\psi_0(z,q) \in \cJ_0$.

    Let $E(z) = -i e^{-i\g z} \psi(z)$ for some $\psi \in \cJ_0$. By Theorem \ref{t1}, there exists
    a unique $q \in \cP$ such that $\psi = \psi_0(\cdot,q)$ and then, due to Proposition \ref{pr1},
    $E$ has form (\ref{p1e11}), where $\vp$ is a solution of (\ref{p1e4}) for this $q$.
\end{proof}

\footnotesize
\no {\bf Acknowledgments.} E. K. is supported by the RSF grant No. 18-11-00032.
D. M. is supported by the RFBR grant No. 19-01-00094.

\end{document}